\documentclass[lettersize,onecolumn]{IEEEtran}
\usepackage{amsmath,amsfonts,amsthm,amssymb}
\usepackage{algorithm}
\usepackage{algpseudocode}
\usepackage{array}
\usepackage[caption=false,font=normalsize,labelfont=sf,textfont=sf]{subfig}
\usepackage{textcomp}
\usepackage{url}
\usepackage{verbatim}
\usepackage{graphicx}
\usepackage{cite}
\usepackage{wrapfig}
\usepackage{dsfont}
\usepackage{tikz}
\usetikzlibrary{positioning}
\newtheorem{definition}{Definition}
\newtheorem{proposition}{Proposition}
\newtheorem{theorem}{Theorem}
\newtheorem{lemma}{Lemma}
\newtheorem{corollary}{Corollary}
\newtheorem{remark}{Remark}

\newif\ifCOM
\COMfalse

\begin{document}

\title{Classification of independent sets in signed Johnson graphs and applications to kissing arrangements}

\author{Rustem~Takhanov and Stanislav~Yun
\thanks{R.~Takhanov and S.~Yun are with the Mathematics Department, Nazarbayev University, and Nazarbayev University Research Administration, Astana, Kazakhstan e-mail: rustem.takhanov@nu.edu.kz.}
}



\maketitle

\begin{abstract}
Johnson graph are a family of graphs that play an important role in the theory of constant-weight codes, extremal combinatorics, and combinatorial geometry. We study signed analogues of classical Johnson graphs, denoted by $J_\pm(n,k)$, whose vertices are vectors of the form $\pm e_{i_1}\pm\cdots\pm e_{i_k}$,
where two vertices are adjacent whenever their dot product equals $k-1$. We are particularly interested in maximum independent sets in the case $k=4$. An example of such an independent set in $J_\pm(n,4)$, which we call \emph{classical}, is obtained by lifting an arbitrary optimal $(n,4,4)$-code.
Such independent sets naturally define kissing arrangements in ${\mathbb R}^n$. Their cardinalities provide the best known lower bounds on the kissing number in dimensions $7$, $9$, and $12$, bounds that in dimensions $9$ and $12$ date back to 1971 and remain unsurpassed today.

We develop an algorithm that is practical for computing all maximum independent sets in $J_\pm(n,4)$ up to signed permutations for $n\le 12$, $n\ne 11$. In addition to obtaining complete lists, we provide structural characterizations of all types of maximum independent sets in these dimensions, excluding $n=5$ and $n=11$.
Our most striking results concern the case $n=12$. We identify $1579$ non-isomorphic maximum independent sets in $J_\pm(12,4)$, all corresponding to non-isometric kissing arrangements of size $840$ in ${\mathbb R}^{12}$. Structurally, $1575$ of these independent sets arise from three different constructions, the rest are liftings of one of four $(12,4,4)$-codes. To our knowledge, this is the first dimension in which such a large diversity of potentially optimal kissing arrangements has been observed.

Beyond this finite range, we prove that for $n\equiv 2$ or $4 \pmod 6$, every maximum independent set arises from a Steiner quadruple system. We also obtain a characterization of the so-called \emph{nontrivially self-compatible} codes, namely optimal $(n,4,4)$-codes from which non-classical maximum independent sets can be constructed.
\end{abstract}

\section{Introduction}
\subsection{The Johnson graph  and codes}

Let us denote $\{1,2,\dots,n\}$ by $[n]$. The set of all $k$-subsets $A\subseteq [n]$ are denoted by $\binom{[n]}{k}$. Given $A\subseteq [n]$, its binary incidence vector is defined as
\[
\mathds{1}_A=(x_1,\dots,x_n)\in \{0,1\}^n,
\qquad
x_i=
\begin{cases}
1,& i\in A,\\
0,& i\notin A.
\end{cases}
\]

\begin{definition}
Let $n,k\in {\mathbb N}$, $n\geq k$. The \emph{Johnson graph} $J(n,k)$ is a simple graph whose vertex set is
\[
V(J(n,k))=\left\{\mathds{1}_A\mid A\in \binom{[n]}{k}\right\}.
\]
The set of its edges is defined by
\[
\{\mathds{1}_A, \mathds{1}_B\}\in E(J(n,k)) \Leftrightarrow |A\cap B|=k-1.
\]
\end{definition}
Johnson graph is a regular, vertex-transitive graph. Its natural generalization is a graph $J(n,k,t)$, with the same set of vertices, but $\{\mathds{1}_A, \mathds{1}_B\}\in E(J(n,k,t)) \Leftrightarrow |A\cap B|=t$. The structure of independent sets in the latter type of graphs is a subject of combinatorial extremal theory~\cite{FRANKL1985160,AHLSWEDE1997125} with applications in combinatorial geometry~\cite{Frankl1981}.

We will be interested in the study of independent sets in $J(n,k)$, motivated by applications in the kissing number problem. In fact, the graph-theoretic problem of finding a largest independent set in $J(n,k)$ is precisely the coding-theoretic problem of finding a largest binary weight-$k$ code with minimum distance $4$. Recall that an \emph{$(n,k,4)$-code} is a set $\mathcal C\subseteq \{0,1\}^n$ such that: (a) every codeword in $\mathcal C$ has $k$ non-zero entries; (b) for all distinct $x,y\in \mathcal C$, the number of entries where $x$ and $y$ differ (Hamming distance), denoted $d(x,y)$, satisfies $d(x,y)\ge 4$.
\begin{proposition}[\cite{bok:MW}~]
A family $\mathcal F$ is an independent set in $J(n,k)$ if and only if $\mathcal F$
is an $(n,k,4)$-code.
\end{proposition}

The symmetric group $S_n$ acts naturally on $[n]$, hence also on the set
of $k$-subsets of $[n]$. Namely, for $\sigma\in S_n$ and
$A=\{a_1,\cdots,a_k\}\in \binom{[n]}{k}$, define
\[
\sigma(\mathds{1}_A)=\mathds{1}_{\sigma(A)}, {\rm where\,\,}\sigma(A)=\{\sigma(a_1),\cdots,\sigma(a_k)\}.
\]
Since permutations preserve cardinalities of intersections, we obtain that every $\sigma\in S_n$ induces an automorphism of $J(n,k)$.

Hence maximal independent sets of $J(n,k)$, equivalently optimal $(n,k,4)$-codes,
are interesting only up to permutation equivalence. Such a classification is known
for $n\le 12$; see Östergård~\cite{10.1109/TIT.2010.2050922}. 

\subsection{The signed Johnson graph $J_{\pm}(n,k)$.}
Let
\[
V\bigl(J_{\pm}(n,k)\bigr)
=
\left\{
x\in \{-1,0,1\}^n \mid \sum_{i=1}^n |x_i|=k
\right\},
\]
that is, all vectors with exactly $k$ nonzero coordinates, each equal to $1, -1$.
Equivalently, one may view a vertex as a $k$-subset $A\subseteq [n]$ together with a choice
of signs on its elements.
For $x,y\in V(J_{\pm}(n,k))$, define adjacency by
\[
\{x, y\}\in E(J_{\pm}(n,k))
\quad\Longleftrightarrow\quad
x^\top y = k-1,
\]
or, equivalently, $x$ and $y$ agree in exactly $k-1$ non-zero coordinates (with the same sign).

\section{Kissing arrangements in ${\mathbb R}^{n}$}

A set of points on the sphere of radius $2$, ${\mathcal K}\subseteq {\mathbb R}^n$, is called a kissing arrangement in $\mathbb R^n$, if the Euclidean distance between any two distinct points is at least $2$. Equivalently, if the points are normalized to lie on the unit sphere, then all pairwise inner products do not exceed $\tfrac{1}{2}$.
The next proposition shows why maximum independent sets in $J_{\pm}(n,k)$ for $k=4$ are of central interest.
\begin{proposition}[Folklore]
Let $I\subseteq V(J_{\pm}(n,4))$ be an independent set, and put
\[
\mathcal K = I \cup \{\pm 2e_1,\dots,\pm 2e_n\}\subset \mathbb R^n.
\]
Then every vector in $\mathcal K$ has norm $2$, and for all distinct $u,v\in \mathcal K$,
\[
u^\top v \le 2.
\]
Thus, ${\mathcal K}$ is a kissing arrangement in $\mathbb R^n$.
\end{proposition}
Observe that every independent set $C\subseteq V(J(n,4))$ canonically induces an independent set in $J_{\pm}(n,4)$: one replaces each codeword by all sign vectors supported on the same $4$-subset of coordinates. In other words, each vector of $C$ is lifted to all $2^4=16$ possible sign combinations on its support. We call a maximum independent set in $J_{\pm}(n,4)$ of this kind {\em classical}.
Remarkably, for $n=9$ and $n=12$, this lifting construction still produces record kissing configurations in $\mathbb R^n$, of cardinalities $306$ and $840$, respectively.  These examples go back to Leech and Sloane~\cite{Leech_Sloane_1971}. For $n=7$ and $n=8$ we obtain kissing arrangements isometric to root systems of type $E_7$ and $E_8$.

The main goal of this paper is to study maximum independent sets in
$J_{\pm}(n,4)$ for $n\leq 12$. Since every signed permutation is an
automorphism of $J_{\pm}(n,4)$, we classify maximum independent sets up
to signed permutations. Our main results are highlighted in bold in
Table~\ref{tab:n44codes}. In addition to determining the number of
equivalence classes under signed permutations, we also provide
structural characterizations of the corresponding independent sets.
\begin{table}[h]
\centering
\begin{tabular}{c c c c c p{7.5cm}}
\hline
$n$
& $A(n)$
& $B(n)$
& $C(n)$
& $D(n)$ & Comment \\
\hline
$5$  & $1$  &   $1$   & $\bf 336$ &  16 & The case $n\equiv 5 \,({\rm mod\,\,} 6)$ usually gives a lot of sets. See subsection~\ref{case5mod6} for details.\\
$6$  & $3$  &   $1$   & $\bf 7$ & 60 & Current record lower bound for the kissing number is 72. Six of 7 maximum independent sets are non-classical. See subsection~\ref{case6} for details.\\
$7$  & $7$  & $1$  & $\bf 2$ & 126 & The record kissing number lower bound. Of the two non-isomorphic maximum independent sets, the classical one corresponds to the root system of type $E_7$ and the non-classical one corresponds to the kissing arrangement found by~\cite{Conway1995}. See subsection~\ref{steiner-triple-leave} for details.\\
$8$  & $14$ & $1$  &  $\bf 1$  & 240 & Optimal and equivalent to the root system of type $E_8$. See subsection~\ref{allsteiner} for details.\\
$9$  & $18$ & $1$  & $\bf 2$ & 306 & The record kissing number lower bound. The two maximum independent sets lead to two non-isomorphic kissing arrangements: the classical one and another one found by~\cite{cohn2026variationsfivedimensionalspherepackings}. \\
$10$ &  $30$  & $1$  & $\bf 1$ & 500 & Current record kissing number lower bound is 510~\cite{GANZHINOV202512}. Only one classical independent set from $S(3,4,10)$.\\
$11$ &  $35$  & $11$  & $\bf ?$  & 582 & Current record kissing number lower bound is 593~\cite{novikov2025alphaevolvecodingagentscientific}. A huge number of non-isomorphic sets beyond our computational resources.\\
$12$ & $51$ & $17$ &  $\bf 1579$   & 840 & The record kissing number lower bound. The first example of a blow-up in a set of possibly optimal kissing arrangements.\\
$14$ & $91$ & $4$ &  $\bf 4$   & 1484 &  Current record is 1932~\cite{GANZHINOV202512}. Classical independent sets are from $S(3,4,14)$.\\
$16$ & $140$ & $1054163 $ &  $\bf 1054163 $   & 2272 &  Current record is 4320~\cite{Barnes_Wall_1959}. Classical independent sets are from $S(3,4,16)$.\\
\hline
\end{tabular}
\caption{
For each $n$, the table contains:
$A(n)$ --- the maximum size of a binary $(n,4,4)$-code;
$B(n)$ --- the number of non-isomorphic optimal binary $(n,4,4)$-codes;
$C(n)$ --- the number of non-isomorphic maximum independent sets in the signed Johnson graph $J_{\pm}(n,4)$;
$D(n)$ --- size of the corresponding kissing arrangement in ${\mathbb R}^n$.
}
\label{tab:n44codes}
\end{table}

\section{Structure of maximum independent sets in $J_{\pm}(n,k)$}
We denote by $\alpha(G)$ the independence number of a graph $G$.
For each sign vector
$
s=(s_1,\dots,s_n)\in\{1,-1\}^n,
$
define
$$
V_s=\{x\in V(J_{\pm}(n,k))\mid \ x_i\in\{0,s_i\}\text{ for all }i\}.
$$
Let $G_s$ be the induced subgraph of $J_{\pm}(n,k)$ on $V_s$.

Then $G_s\cong J(n,k)$. Indeed, every $x\in V_s$ is determined uniquely by its support
$
\operatorname{supp}(x)=\{i:x_i\neq 0\}\in \binom{[n]}k,
$
and for $x,y\in V_s$ we have
$
x\cdot y = |\operatorname{supp}(x)\cap \operatorname{supp}(y)|.
$
Hence $x$ and $y$ are adjacent in $G_s$ exactly when their supports intersect in $k-1$ points, which is precisely the adjacency relation in $J(n,k)$.

Thus the images of $J(n,k)$ under signed coordinate permutations are exactly the induced subgraphs $G_s$, $s\in\{1,-1\}^n$. There are
$
N=2^n
$
such subgraphs.  The edge-counting argument used in the following theorem traces back to Lemma~1 of~\cite{10.1007/978-94-010-1826-5_11}.

\begin{theorem}[Main fact]\label{main-fact}
We have $\alpha\bigl(J_{\pm}(n,k)\bigr) = 2^k\,\alpha\bigl(J(n,k)\bigr)$, and if $I$ is a maximum independent set in $J_{\pm}(n,k)$ (i.e. an independent set of size
$2^k\,\alpha\bigl(J(n,4)\bigr)$), 
then for every $s\in\{1,-1\}^n$,
$
|I\cap V_s|=\alpha\bigl(J(n,k)\bigr).
$
\end{theorem}
\begin{proof} 
Let us introduce a new bipartite graph $\mathcal B$ with parts
$V\bigl(J_{\pm}(n,k)\bigr)$
and  $\{1,-1\}^n$ where a vertex $x\in V(J_{\pm}(n,k))$ is joined to $s$ if and only if $x\in V_s$.

Since $G_s\cong J(n,k)$, its vertex set has size
$
|V_s|=\binom nk.
$
So every right-side vertex $s$ has degree $\binom nk$.

Fix $x\in V(J_{\pm}(n,k))$. Its support has size $k$. In order that $x\in V_s$, the sign vector $s$ must agree with $x$ on those $k$ coordinates, while on the remaining $n-k$ coordinates $s$ is arbitrary. Therefore $x$ belongs to exactly
$
2^{n-k}
$
of the subgraphs $G_s$.

Let
$
I\subseteq V\bigl(J_{\pm}(n,k)\bigr)
$
be an independent set. From the previous analysis of degrees in $\mathcal B$ it follows that the total number of edges in $\mathcal B$ that has one endpoint in $I$, is $2^{n-k}|I|$. On the other hand, for each $s\in\{1,-1\}^n$, the intersection
$
I_s= I\cap V_s
$
is an independent set in $G_s\cong J(n,k)$. Hence
$$
|I_s|\le \alpha\bigl(J(n,k)\bigr)
\qquad\text{for all }s.
$$
In other words, the total number of edges in $\mathcal B$ that has one of endpoints at $s$ and another from $I$ is no more than $\alpha\bigl(J(n,k)\bigr)$. 
Summing over all $s\in\{1,-1\}^n$, we obtain
$
\sum_{s\in\{1,-1\}^n} |I\cap V_s|
\le
2^n\,\alpha\bigl(J(n,k)\bigr).
$
Therefore, by counting edges in $\mathcal B$,
$$
2^{n-k}|I|\le 2^n\,\alpha\bigl(J(n,k)\bigr),
$$
hence
$
|I|\le 2^k\,\alpha\bigl(J(n,k)\bigr).
$

But for any independent set in $J(n,k)$ one can correspond an independent set in $J_{\pm}(n,k)$ by adding all possible sign combinations to coordinates of its elements. So, an independent set of size $2^k\,\alpha\bigl(J(n,k)\bigr)$ exists in $J_{\pm}(n,k)$ and we conclude that $\alpha\bigl(J_{\pm}(n,k)\bigr) = 2^k\,\alpha\bigl(J(n,k)\bigr)$. This proves the first part of theorem.

The second part basically says that any extremal independent set in $J_{\pm}(n,k)$ must intersect \emph{every} induced copy $G_s\cong J(n,k)$ in a maximum independent set.
Now let us assume that $I$ is a maximum independent set. As above,
$
|I\cap V_s|\le \alpha\bigl(J(n,k)\bigr)
$
for all $s$, and $\sum_{s\in\{1,-1\}^n} |I\cap V_s|$ is the number of edges in $\mathcal B$ with one endpoint from $I$.
Therefore,
$
\sum_{s\in\{1,-1\}^n} |I\cap V_s|
=
2^{n-k}|I|.
$

Using $|I|=2^k\,\alpha(J(n,k))$, we get
$
\sum_{s\in\{1,-1\}^n} |I\cap V_s|
=
2^{n-k}\cdot 2^k\,\alpha\bigl(J(n,k)\bigr)
=
2^n\,\alpha\bigl(J(n,k)\bigr).
$

Since each term satisfies
$
|I\cap V_s|\le \alpha\bigl(J(n,k)\bigr),
$
we must have equality termwise, i.e.
$
|I\cap V_s|=\alpha\bigl(J(n,k)\bigr)$
for all $s$.
\end{proof}
We define the
\emph{sign-forgetting map}
\[
\operatorname{SF}:V(J_\pm(n,k))\to V(J(n,k)),\qquad
x\mapsto \mathds{1}_{\operatorname{supp}(x)}.
\]
Theorem~\ref{main-fact} says that if $I\subset V(J_\pm(n,4))$ is a maximum
independent set, then $\operatorname{SF}(I\cap V_s)$ is an optimal
$(n,4,4)$-code.
\subsection{Compatibility of $(n,k,4)$-codes}
Let $I$ be a maximum independent set in $J_{\pm}(n,k)$ and let us treat it as a set of vectors in ${\mathbb R}^n$.
Theorem~\ref{main-fact} basically claims that the intersection of $I$ with any hyperoctant is an optimal $(n,k,4)$-code (or, a maximum independent set in $J(n,k)$) modulo signs. This restriction is quite rigid, i.e. not any pair of $(n,k,4)$-codes can be occupy neighbouring hyperorthants. Let us study that issue.

For each $m \in \{1,\dots,n\}$, define a map
$f_m : \mathbb{R}^n \to \mathbb{R}^n$
by
\[
f_m(x_1,\dots,x_n) = (x_1,\dots,x_{m-1}, -x_m, x_{m+1},\dots,x_n),
\]
that is, $f_m$ changes the sign of the $m$-th coordinate.
\begin{lemma}\label{compatibility-simple}
Let $C,C'$ be maximum independent sets in $J(n,k)$, let $m\in [n]$. 
Then
$C\cup f_m(C')$
is independent in $J_\pm (n,k)$ if and only if
\[
\{S\mid \mathds{1}_S\in C', m\notin S\}=\{S \mid  \mathds{1}_S\in C, m\notin S\}.
\]
\end{lemma}
\begin{proof}
$(\Rightarrow)$ Assume that $C\cup f_m(C')$
is independent. 

Let $\mathds{1}_S\in C$ with $m\notin S$, and suppose that $\mathds{1}_S\notin C'$.
Since $C'$ is a maximum independent set in $J(n,k)$, it is maximal, so there exists
$\mathds{1}_A\in C'$ adjacent to $\mathds{1}_S$, that is,
\[
|S\cap A|=k-1.
\]
As $m\notin S$, we also have $m\notin S\cap A$, and thus,
\[
\langle \mathds{1}_S,f_m(\mathds{1}_A)\rangle = k-1,
\]
contrary to the assumption that $C\cup f_m(C')$ is an independent set.
Hence every $\mathds{1}_S\in C$ with $m\notin S$ belongs to $C'$.
By symmetry, every $\mathds{1}_S\in C'$ with $m\notin S$ belongs to $C$.
This proves $\{S\mid \mathds{1}_S\in C , m\notin S\}=\{S \mid \mathds{1}_S\in C', m\notin S\}$.

\noindent
$(\Leftarrow)$ Assume now that
$\{S \mid \mathds{1}_S\in C, m\notin S\}=\{S\mid \mathds{1}_S\in C', m\notin S\}$.

Let us define $H=\{\mathds{1}_S\in C \mid m\notin S\}=\{\mathds{1}_S\in C' \mid m\notin S\}$,
and
$A = C\setminus H$, $B = C'\setminus H$,
where every block in $A\cup B$ contains $m$.
Let us show that $C\cup f_m(C')$
is an independent set in $J_\pm (n,k)$.

First, inside $C$ all inner products are at most $2$, since $C$ is an independent
set in $J(n,k)$. The same holds inside $f_m(C')$, because $f_m$ preserves inner products.
It remains to check mixed inner products between any vector from $C$ and any vector from $f_m(C')$.

Let $x\in C$ and $y\in C'$.
\begin{itemize}
\item The case of $x,y\in H$ is within $C'$;

\item If $x\in H$ and $y\in B$, then $x,y\in C'$, hence $|x^\top y|\le k-2$ and
$x_m=0$, so $\langle x, f_m(y)\rangle = |x^\top y| \le k-2$.
If $x\in A$ and $y\in H$, the argument is symmetric.

\item If $x\in A$ and $y\in B$, then $x_m=y_m=1$, hence
$\langle x, f_m(y)\rangle = |x^\top y| - 2 \le k-2$.
\end{itemize}
Thus every pair of distinct vectors in $C\cup f_m(C')$ has inner product at most $k-2$,
so this set is an independent set.
\end{proof}
Lemma~\ref{compatibility-simple} has the following corollary which will be instrumental in the computational part of the paper.
\begin{corollary}\label{compatibility}
Let $C,C'$ be maximum independent sets in $J(n,k)$, let $m\in [n]$, and let
$\pi\in S_n$. Set
$v=\pi^{-1}(m)$. 
Then
$C\cup f_m(\pi(C'))$
is independent in $J_\pm (n,k)$ if and only if the restriction
\[
\pi|_{[n]\setminus\{v\}} : [n]\setminus\{v\}\to [n]\setminus\{m\}
\]
is an isomorphism of hypergraphs
\[
\bigl([n]\setminus\{v\},\,\{S\mid \mathds{1}_S\in C', v\notin S\}\bigr)
\quad\text{and}\quad
\bigl([n]\setminus\{m\},\,\{S \mid \mathds{1}_S\in C, m\notin S\}\bigr).
\]
\end{corollary}

\begin{proof}
Define
$D=\pi(C')$.
Since $\pi$ is a permutation of coordinates, $D$ is again a maximum independent set
in $J(n,k)$. Moreover, because $v=\pi^{-1}(m)$, for every $\mathds{1}_T\in C'$ the corresponding block $T$ satisfies
$m\in \pi(T)\iff v\in T$. 
Therefore,
\begin{equation}\label{vk-change}
\{S\mid \mathds{1}_S\in D, m\notin S\}
=
\{\pi(T)\mid \mathds{1}_T\in C',\ v\notin T\}.
\end{equation}

By Lemma~\ref{compatibility-simple}, $C\cup f_m(D)$
is independent if and only if
\[
\{S\mid \mathds{1}_S\in C , m\notin S\}=\{S \mid \mathds{1}_S\in D, m\notin S\},
\]
that is, if and only if 
\[
\{S\mid \mathds{1}_S\in C , m\notin S\}=\{\pi(T)\mid \mathds{1}_T\in C',\ v\notin T\}.
\]
But the latter exactly means that $\pi|_{[n]\setminus\{v\}}$ maps the collection of sets
$\{S\mid \mathds{1}_S\in C', v\notin S\}$ to $\{S \mid \mathds{1}_S\in C, m\notin S\}$. Since $\pi|_{[n]\setminus\{v\}} : [n]\setminus\{v\}\to [n]\setminus\{m\}$ is one-to-one, this is an isomorphism of hypergraphs. 
\end{proof}

\begin{definition}
An optimal $(n,4,4)$-code $C$ is called
\emph{egotistic} if for any optimal $(n,4,4)$-code $C'$ and any $m\in [n]$, $C\cup f_m(C')$ being independent in $J_\pm (n,4)$ implies $C=C'$.
\end{definition}

\begin{theorem}\label{classical-steiner} Suppose $n\equiv 2{\rm\,\,or\,\,}4\,({\rm mod\,\,}6)$, i.e. there is a Steiner system with parameters $(3,4,n)$. Then, any maximum independent set in $J_\pm (n,4)$ is classical with a base code being a quadruple Steiner system over $[n]$. 
\end{theorem}
To prove the theorem we need the following lemma. It implies that all quadruple Steiner systems are, actually, egotistic. 
\begin{lemma}\label{steiner} Suppose that $C$ and $C'$ are both Steiner systems with parameters $(3,4,n)$. Then, $\{\mathds{1}_S\mid S\in C\}\cup f_m(\{\mathds{1}_S\mid S\in C'\})$ is independent in $J_\pm (n,4)$ if and only if $C=C'$.
\end{lemma}
\begin{proof} From $C=C'$ and the fact that Steiner systems with parameters $(3,4,n)$ define $(n,4,4)$-codes, the independence of $\{\mathds{1}_S\mid S\in C\}\cup f_m(\{\mathds{1}_S\mid S\in C'\})$ is straightforward. Now let us prove the opposite statement. 

From Lemma~\ref{compatibility-simple} we obtain that the sets
$\{S\in C\mid m\notin S\}$
 and 
$\{S\in C'\mid m\notin S\}$
coincide. Let
$H=
\{S\in C\mid m\notin S\}
=
\{S\in C'\mid m\notin S\}$.
Let us show that the blocks containing $m$ are uniquely determined by $H$.

Let $T\subseteq [n]\setminus\{m\}$ be any $3$-subset. Since $C$ is a
Steiner system $S(3,4,n)$, the triple $T$ is contained in a unique block
of $C$.

There are two possibilities:
\begin{itemize}
    \item The unique block containing $T$ does not contain $m$. Then
    $T$ is contained in some block of $H$.

    \item The unique block containing $T$ contains $m$. Since every block
    has size $4$, this block must be exactly
    $
    T\cup\{m\}.
    $
\end{itemize}

Therefore,
\[
T\cup\{m\}\in C
\quad\Longleftrightarrow\quad
T \text{ is not contained in any block of } H.
\]

Hence the collection of blocks of $C$ containing $m$ is completely
determined by $H$. The same argument applies to $C'$.
Since $C$ and $C'$ have the same blocks avoiding $m$, they also have the
same blocks containing $m$. Consequently,
$C=C'$.
\end{proof}
\begin{proof}[Proof of Theorem~\ref{classical-steiner}] 
Now assume that a Steiner system $S(3,4,n)$ exists. Let $C$ be an optimal $(n,4,4)$-code. No $3$-subset of $[n]$ can belong to two different blocks of
$C$, therefore $|C|\le \frac{\binom{n}{3}}{4}$. A Steiner system has
exactly $\frac{\binom{n}{3}}{4}$
blocks, so the above bound is sharp. Hence every optimal $(n,4,4)$-code must satisfy
$|C|=\frac{\binom{n}{3}}{4}$. Since distinct blocks contribute disjoint collections of triples, and
the total number of covered triples equals the total number of triples
in $[n]$, every triple must belong to exactly one block of $C$.
Thus, $C$ is itself a Steiner system $S(3,4,n)$.

Lemma~\ref{steiner} implies the following conclusion: if $n$ is such that there is a Steiner system with parameters $(3,4,n)$, then for any maximum independent set in $J_\pm (n,4)$, adjacent hyperoctants of ${\mathbb R}^n$ must be occupied by the same code $C$ (corresponding to a Steiner system). So, this implies that this maximum independent set is the lifted code $C$ and is classical. 
\end{proof}
\begin{remark}
The property of being egotistic is not equivalent to being a Steiner
system. Indeed, among the $17$ non-isomorphic optimal
$(12,4,4)$-codes classified by Best~\cite{Best1978}, the codes numbered $14$-$17$ are egotistic, although they are not Steiner systems. Thus egotisticity should be understood as a rigidity
property with respect to signed liftings in $J_\pm(n,4)$, rather than
as a design-theoretic characterization of Steiner quadruple systems.

The examples above suggest an arithmetic taxonomy. For
$n\equiv 0\pmod 6$, egotistic optimal codes do occur. For
$n\equiv 2,4\pmod 6$, every optimal $(n,4,4)$-code is a Steiner
system $S(3,4,n)$, and hence is egotistic by
Lemma~\ref{steiner}. The most interesting case therefore appears to be
$n\equiv 1,3\pmod 6$,
when the leave of an optimal $(n,4,4)$-code is a Steiner triple
system $S(2,3,n)$.

The optimal $(7,4,4)$- and $(9,4,4)$-codes are
non-egotistic (proofs can be found in subsection~\ref{steiner-triple-leave}). Moreover, all optimal $(13,4,4)$-codes checked by the
authors were also non-egotistic. This evidence suggests that the
presence of a Steiner-triple-system leave may systematically produce
non-classical liftings. Although we do not currently have a proof, we
would not be surprised if every optimal $(n,4,4)$-code with
$n\equiv 1,3\pmod 6$ were non-egotistic.
\end{remark}
Let us now show that existence of any two unequal compatible optimal $(n,4,4)$-codes leads to a non-classical independent set.
\begin{lemma}\label{conway}
Let $C,C'$ be two optimal $(n,4,4)$-codes such that $C\ne C'$.
Assume that there exists $m\in[n]$ such that
\[
\{S\mid \mathds{1}_S\in C, m\notin S\}=\{S \mid  \mathds{1}_S\in C', m\notin S\}.
\]
Define
\[
I_+
=
\{(\sigma_1x_1,\ldots,\sigma_nx_n):
x\in C,\ \sigma_i\in\{-1,1\}\text{ for }i\neq m,\ \sigma_m=1\},
\]
and
\[
I_-
=
\{(\sigma_1x_1,\ldots,\sigma_nx_n):
x\in C',\ \sigma_i\in\{-1,1\}\text{ for }i\neq m,\ \sigma_m=-1\}.
\]
Then, $I_+\cup I_-$
is a non-classical maximum independent set in $J_{\pm}(n,4)$.
\end{lemma}

\begin{proof}
Since $I_+$ and $I_-$ are independent as subsets of classical maximum independent sets,
it remains to check that there are no edges between $I_+$ and $I_-$.
Take $x\in I_+$, $y\in I_-$,
and let $S=\operatorname{supp}(x)\in C$, $T=\operatorname{supp}(y)\in C'$.
Suppose, for contradiction, that $x$ and $y$ are adjacent in
$J_{\pm}(n,4)$. Then $x^\top y=3$.
Since both $x$ and $y$ have exactly four nonzero coordinates, this implies
that $|S\cap T|=3$
and that the signs of $x$ and $y$ coincide on $S\cap T$.

If $m\in S\cap T$, then $x_m=1$ and $y_m=-1$, so the signs do not
coincide on $m$. Therefore $m\notin S\cap T$.
If $m\notin S$ and $m\notin T$, then by the assumption on the restrictions
outside $m$, both $S$ and $T$ belong to $C$. Hence
$|S\cap T|=3$ contradicts the fact that $C$ is an $(n,4,4)$-code.
If $m\in S$ and $m\notin T$, then $T\in C$ by the same assumption.
Thus $S,T\in C$, and again $|S\cap T|=3$ is impossible. If $m\notin S$ and $m\in T$, then $S\in C'$. Thus $S,T\in C'$, and
$|S\cap T|=3$ contradicts the fact that $C'$ is an $(n,4,4)$-code.
Therefore no edge joins $I_+$ and $I_-$. Hence $I_+\cup I_-$ is independent.

Since $C$ and $C'$ are both optimal, we  have $|C|=|C'|$.
Let us denote
\[
C_0=\{S\mid \mathds{1}_S\in C, m\notin S\}=\{S \mid  \mathds{1}_S\in C', m\notin S\},
\]
and
\[
C_1=\{S\mid \mathds{1}_S\in C, m\in S\},
\qquad
C'_1=\{S \mid  \mathds{1}_S\in C', m\in S\}.
\]
Then $|C_1|=|C'_1|$, due to
\[
|C_0|+|C_1|=|C|=|C'|=|C_0|+|C'_1|.
\]

For every support not containing $m$, both $I_+$ and $I_-$ contain all
$2^4$ sign choices, so these vertices are counted only once in the union.
For every support containing $m$, $I_+$ contains the $2^3$ sign choices
with $x_m=1$, while $I_-$ contains the $2^3$ sign choices with
$x_m=-1$. Therefore,
\[
|I_+\cup I_-|
=
2^4|C_0|+2^3|C_1|+2^3|C'_1|.
\]
Since $|C_1|=|C'_1|$, we get
\[
|I_+\cup I_-|
=
2^4|C_0|+2^4|C_1|
=
2^4|C|,
\]
which is the independence number of $J_{\pm}(n,4)$.
Thus $I_+\cup I_-$ is a maximum independent set in $J_{\pm}(n,4)$.
\end{proof}

The following property is slightly strongly than non-egotism.
\begin{definition}
An optimal $(n,4,4)$-code $C$ is called
\emph{nontrivially self-compatible (NSC)} if there exists a maximum
independent set $I\subset V(J_\pm(n,4))$
such that:
\begin{itemize}
\item for every $s\in\{1,-1\}^n$,
after we apply sign-forgetting to $I\cap V_s$, $\operatorname{SF}(I\cap V_s)=\pi_s(C)$
for some permutation
$\pi_s\in S_n$,
possibly depending on $s$, and $\pi_{(1,...,1)} = {\rm id}$;
\item there exists at least one $\tau\in\{1,-1\}^n$ such that
$\pi_\tau(C)
\neq C$.
\end{itemize}
\end{definition}
Let ${\rm Aut}(C)$ denote a set of permutations of $[n]$ such that $\pi(C)
= C$.
\begin{theorem} An optimal $(n,4,4)$-code $C$ is NSC if and only if there is $m\in [n]$ and 
$\pi\in S_n\setminus {\rm Aut}(C)$ such that
\[
\pi|_{[n]\setminus\{v\}} : [n]\setminus\{v\}\to [n]\setminus\{m\},
\]
where $v=\pi^{-1}(m)$,
is an isomorphism of hypergraphs
\[
\bigl([n]\setminus\{v\},\,\{S\mid \mathds{1}_S\in C, v\notin S\}\bigr)
\quad\text{and}\quad
\bigl([n]\setminus\{m\},\,\{S \mid \mathds{1}_S\in C, m\notin S\}\bigr).
\]
\end{theorem}

\begin{proof}
Assume that $C$ is NSC. Then there exists a
maximum independent set
$I\subset V(J_{\pm}(n,4))$
such that, for every hyperoctant $V_s$,
the sign-forgetting projection of $I\cap V_s$ is permutation-equivalent to $C$, and for at least one hyperoctant
$V_\tau$, the sign-forgetting projection of $I\cap V_\tau$ is $\pi(C)\neq C$
for some permutation $\pi\in S_n\setminus {\rm Aut}(C)$.

Applying signed permutations if necessary, we may assume that $I\cap V_{+1\cdots+1}=C$,
 and that
the sign-forgetting projection of  $I\cap H_-$ is $\pi(C)$,
where $H_-$ is the hyperoctant obtained from $V_{+1\cdots+1}$ by changing the sign
of the $m$-th coordinate for some $m\in[n]$. Thus
$H_-=f_m(H_+)$.

Since $I$ is independent in $J_{\pm}(n,4)$, the set
$C\cup f_m(\pi(C))$
is independent in $J_{\pm}(n,4)$. Hence, by
Corollary~\ref{compatibility},
the restriction
$\pi|_{[n]\setminus\{v\}}$
with $v=\pi^{-1}(m)$,
is an isomorphism between
\[
\bigl([n]\setminus\{v\},\,\{S\mid \mathds{1}_S\in C,\ v\notin S\}\bigr)
\]
and
\[
\bigl([n]\setminus\{m\},\,\{S\mid \mathds{1}_S\in C,\ m\notin S\}\bigr).
\]

Conversely, suppose that there exist
$m\in[n]$ and
$\pi\in S_n\setminus {\rm Aut}(C)$
such that
$\pi|_{[n]\setminus\{v\}}$
with $v=\pi^{-1}(m)$,
is an isomorphism of the above hypergraphs.

Set
$C'=\pi(C)$.
Then
$C'\neq C$,
since
$\pi\notin {\rm Aut}(C)$.
By Corollary~\ref{compatibility},
the set
$C\cup f_m(C')$
is independent in $J_{\pm}(n,4)$. Since both $C$ and $C'$ are optimal
$(n,4,4)$-codes, each contributes
$2^{n-1}|C|$
vertices after all sign choices outside the fixed coordinate.
Hence the construction of Lemma~\ref{conway} yields a maximum independent set
$I=I_+\cup I_-
\subset V(J_{\pm}(n,4))$.

The positive hyperoctant of $I$ projects to $C$, while the hyperoctant
with negative $m$-th coordinate projects to
$C'=\pi(C)\neq C$.
Therefore $I$ is a non-classical lifting of $C$, and thus $C$ is
nontrivially self-compatible.
\end{proof}
\begin{remark} Examples of NSC codes are optimal $(7,4,4)$ and $(9,4,4)$ codes. This shown in subsection~\ref{steiner-triple-leave}.
\end{remark}
\begin{remark}
It is tempting to conjecture that every optimal $(n,4,4)$-code that is not NSC must be egotistic. However, this is not the case.
We checked that among $11$ optimal $(11,4,4)$-codes given in~\cite{pub:7433}, the code $C_2$ is not NSC. On the other hand, $C_2$ is compatible with the distinct code $C_1$ from the same list. Consequently, $C_2$ is not egotistic.
Thus, nontrivial self-compatibility and non-egotisticity are independent phenomena. A code may fail to admit a nontrivial octant-wise self-gluing while still being compatible with another optimal code.
\end{remark}
\section{Listing of independent sets and decoy generation}\label{computer-aided}
Note that $J_\pm(n,k)$ admits all signed permutations as automorphisms. Our goal is to list all maximum independent sets in $J_\pm(n,k)$ modulo a signed permutation. This naturally leads to the following definition.
\begin{definition}
A family of independent sets $\{S_1,\dots,S_t\}$ in $J_{\pm}(n,k)$ is called a system of decoys if for every maximum independent set $I$ of $J_{\pm}(n,k)$, there exist $l \in [t]$ and a signed permutation $\pi$ such that
\[
\pi(S_l) \subseteq I.
\]
In other words, every maximum independent set contains, up to symmetry, one of the decoys.
\end{definition}
We use Algorithm~\ref{land} to compute all maximum independent sets in $J_{\pm}(n,4)$ up to signed permutations, starting from a given system of decoys. Note that our approach relies on the MOSEK's implementation~\cite{mosek} of the branch-and-bound method of Land and Doig~\cite{c317b127-1fe2-3862-8218-ad5f3daf32e7}, rather than on maximal clique enumeration methods in the style of Östergård~\cite{OSTERGARD2002197}. This choice is motivated by empirical observations: the Land--Doig method performs significantly more efficiently in our setting than clique-listing approaches.
We attribute this behavior to the fact that the signed Johnson graph $J_{\pm}(n,4)$ is close to being perfect, in the sense that the number of induced odd holes and antiholes is relatively small. Consequently, linear programming relaxations provide strong upper bounds, as is typical for perfect or nearly perfect graphs. These tight bounds substantially enhance the efficiency of the branch-and-bound procedure.

\begin{algorithm}
\small
\caption{Enumeration of all maximum independent sets compatible with a system of decoys}\label{land}
\begin{algorithmic}
\State {\bf Input:} A system of decoys $\mathcal D=\{S_1,\dots,S_t\}$, where each $S_i$ is an independent set in $J_{\pm}(n,4)$.
\State {\bf Output:} Independent sets $I\subseteq V(J_{\pm}(n,4))$ such that $S_i\subseteq I$ for some $i\in\{1,\dots,t\}$ and $|I|=16\,\alpha\bigl(J(n,4)\bigr)$.

\State Initialize $\mathcal F:=\varnothing$.

\For{each decoy $S\in\mathcal D$}
    \State Set $D:=S$.

    \State $U:=\{v\in V(J_{\pm}(n,4))\setminus D \mid \forall w\in D,\ \{v,w\}\notin E(J_{\pm}(n,4))\}$.

    \State $H:=J_{\pm}(n,4)[U]$.

    \State Solve the Integer Linear Programming (ILP) problem
    \[
    \max \sum_{v\in U} x_v
    \]
    \[
    x_u+x_v\leq 1 \qquad \text{for all } \{u,v\}\in E(H),\,\, x_v\in\{0,1\} \qquad \text{for all } v\in U.
    \]
     \hspace{19pt}by the branch-and-bound method.

    \While{the optimum value equals $16\,\alpha\bigl(J(n,4)\bigr)-|D|$}
        \State Let
        \[
        T:=\{v\in U\mid x_v=1\}.
        \]
        \State Record the independent set $I:=D\cup T$ in $\mathcal F$.
        \State Add the exclusion cut
        \[
        \sum_{v\in T} x_v \leq 16\,\alpha\bigl(J(n,4)\bigr)-|D|-1.
        \]
        \State Re-solve the ILP.
    \EndWhile
\EndFor

\State \Return $\mathcal F$.
\end{algorithmic}
\end{algorithm}

Let us now describe how we create a system of decoys. Let $C_1,\dots,C_c \subseteq \binom{[n]}{4}$ be representatives of all $(n,4,4)$-codes up to permutation equivalence. Then the family $\{C_1,\dots,C_c\}$ forms a system of decoys in $J_{\pm}(n,4)$. Indeed, all vectors in the positive hyperoctant of every maximum independent set $I$ in $J_{\pm}(n,4)$ correspond to a code $C \subseteq \binom{[n]}{4}$ (Theorem~\ref{main-fact}). Since $C$ is permutation-equivalent to one of the $C_i$, it follows that $I$ contains a signed copy of some $C_i$. The latter means that when searching for maximum independent sets up to signed permutation equivalence, we may assume that the positive hyperoctant is occupied by one of the codes $C_i \subseteq \binom{[n]}{4}$.

However, as experiments show, fixing the configuration in the positive hyperoctant alone does not sufficiently constrain the structure of a maximum independent set $I$. So, it is necessary to enlarge decoys to include configurations from neighboring hyperoctants (particularly those with one negative coordinate). This expansion radically improves the performance of Algorithm~\ref{land}.

So, let us assume that a decoy consists of the base code $C_i$ together with additional vectors placed in neighbouring hyperoctants.
Each neighbouring hyperoctant is indexed by a coordinate $u\in[n]$, and vectors placed in it corresponds to applying the map $f_u$ onto some $(n,4,4)$-code, or equivalently, as a set $f_u(h(C_j))$ for some permutation $h\in S_n$ and a base code $C_j$.

To determine which vectors may be placed in these neighbouring hyperoctants, we use Corollary~\ref{compatibility}. For each coordinate $u$ and each input code $C_j$, we consider all ways of embedding $C_j$ into the $u$-th hyperoctant so that it remains compatible with the base code $C_i$. This is achieved by comparing the hypergraphs of $C_i$ and $C_j$ (with deleted chosen vertices) and enumerating all isomorphisms between them. Each such isomorphism yields a permutation $h$, and hence a candidate configuration $f_u(h(C_j))$ in the $u$-th hyperoctant.

Since many of these permutations produce identical images of $C_j$, we quotient out this redundancy and retain only one representative for each distinct image. This results in a reduced set of options for each coordinate $u$.

A decoy is then obtained by selecting, independently for each $u\in[n]$, one such configuration and taking the union
\[
C_i \;\cup\; \bigcup_{u=1}^{n} f_u\bigl(h_u(C_{j_u})\bigr).
\]
If this union forms an independent set, we record it as a valid partial solution. The collection of all such decoys is denoted by $\mathcal D$.

Finally, we remove symmetry duplicates. Since we fixed the base code $C_i$, the remaining symmetry group is its stabilizer $\mathrm{Stab}(C_i)$. We therefore identify decoys that differ by an element of $\mathrm{Stab}(C_i)$ and retain one representative from each orbit. Formally, all the details can be found in Algorithm~\ref{decoys}.
\begin{algorithm}
\small
\caption{System of decoys generation}\label{decoys}
\begin{algorithmic}
\Require Codes $C_1,\dots,C_c\subseteq \binom{[n]}{4}$ and generators of $\mathrm{Stab}(C_i)\le S_n$ for all $i\in[c]$
\Ensure A reduced list of decoy representatives

\State $\mathcal D\gets \emptyset$

\For{$1\le i\le j\le c$}
    \For{$u,v\in[n]$}
        \If{the hypergraphs $\bigl([n]\setminus\{u\},\{S\in C_i:u\notin S\}\bigr)$ and $\bigl([n]\setminus\{v\},\{S\in C_j:v\notin S\}\bigr)$ are isomorphic}
            \State compute all isomorphisms
            \[
            h:[n]\setminus\{v\}\to [n]\setminus\{u\}
            \]
            \For{each such $h$}
                \State extend $h$ to a permutation of $[n]$ by setting $h(v)=u$
                \State add $h$ to $I_{i,j,u,v}$
            \EndFor
            \State set $I_{j,i,v,u}=\{h^{-1}:h\in I_{i,j,u,v}\}$
        \EndIf
    \EndFor
\EndFor

\For{$i,j\in[c]$ and $u\in[n]$}
    \State $I_{i,j,u}\gets \bigcup_{v\in[n]} I_{i,j,u,v}$
    \State partition $I_{i,j,u}$ by the relation
    \[
    h\sim h' \iff h(C_j)=h'(C_j)
    \]
    \State choose one representative from each class and denote the resulting set by $E_{i,j,u}$
\EndFor

\For{$i\in[c]$}
    \For{$u\in[n]$}
        \State $S(i,u)\gets \bigcup_{j\in[c]}\{(j,h):h\in E_{i,j,u}\}$
    \EndFor

    \For{each tuple $\bigl((j_1,h_1),\dots,(j_n,h_n)\bigr)\in \prod_{u=1}^{n} S(i,u)$}
        \State form
        \[
        D = C_i\cup \bigcup_{u=1}^{n} f_u\bigl(h_u(C_{j_u})\bigr)
        \]
        \If{$D$ is a kissing arrangement}
            \State add $D$ to $\mathcal D$
        \EndIf
    \EndFor
\EndFor

\State $\mathrm{Reps}\gets \emptyset$

\For{$i\in[c]$}
    \State let $\mathcal D_i\subseteq \mathcal D$ be the subfamily of decoys with base code $C_i$
    \State generate $\mathrm{Stab}(C_i)$ from the given generators
    \State partition $\mathcal D_i$ into orbits under the action of $\mathrm{Stab}(C_i)$
    \State choose one representative from each orbit and add it to $\mathrm{Reps}$
\EndFor

\State \Return $\mathrm{Reps}$
\end{algorithmic}
\end{algorithm}

After generating the system of decoys and enumerating all compatible maximum
independent sets, one still has to factorize the resulting list with respect
to signed permutations. This step is necessary because the same class of independent
sets may appear multiple times during the search procedure. Note that the
corresponding kissing arrangements are non-isometric whenever the underlying
independent sets are inequivalent under signed permutations. Indeed, for
$n\ge 9$, equivalence of maximum independent sets in $J_{\pm}(n,4)$ under
signed permutations is equivalent to isometricity of the associated kissing
arrangements (see Section~\ref{Isomorphism} for details). Consequently, factorization by signed permutations also yields
a classification of the resulting kissing arrangements up to isometry.

\section{The case of $n=12$}
For $v\in \mathbb{R}^n$, let ${\rm supp}(v) = \{i: v_i \neq 0\}$. 
For $A\subseteq [n]$, let us denote 
\[ 
T(A) = \{(x_1,...,x_n) \mid x_i \in\{+1, -1\} \text{ if } i\in A \text{ and } x_i =0 \text{ otherwise}\},
\]
and $T(C)=\cup_{v\in C}T({\rm supp}(v))$.

\begin{definition} 
A subset $S\subseteq [n]$ is called a free zone of a binary $(n,4,4)$-code $C$ (possibly, not optimal), if
$$
C = \{x\in C\mid {\rm supp}(x)\subseteq S\}\cup \{x\in C\mid |{\rm supp}(x)\cap S|\leq 2\}.
$$
\end{definition}
For $S\subseteq [n]$, define $J_\pm (S,4)$ as the induced graph $J_\pm (n,4)[\bigcup_{A\in {S \choose 4}}T(A)]$. Non-classical kissing arrangements in ${\mathbb R}^{12}$ can be constructed using the next lemma whose proof is straightforward.
\begin{lemma}\label{plug-in-J64} Let $S_1,\cdots,S_k\subseteq [n]$ be mutually disjoint free zones of a binary $(n,4,4)$-code. Then, for any  $\{I_j\}_{j=1}^k$ where $I_j$ is an independent set in $J_\pm (S_j,4)$,
$$
\bigcup_{j=1}^k I_j\cup \{T({\rm supp}(x))\mid x\in C, \forall j, {\rm supp}(x)\not\subseteq S_j\}
$$
is an independent set in $J_\pm (n,4)$.
\end{lemma}
\noindent {\bf Construction 1.} The description of an example of a binary $(12,4,4)$-code with a free zone can be found in~\cite{Kalbfleisch1968131}. For completeness, let us reproduce that construction here. We identify $[12]$ with a union of $\{1,2,3,4,5,6\}$ and $\{1',2',3',4',5',6'\}$. 

\begin{wrapfigure}{r}{0.25\textwidth}
\centering
\includegraphics[width=0.2\textwidth]{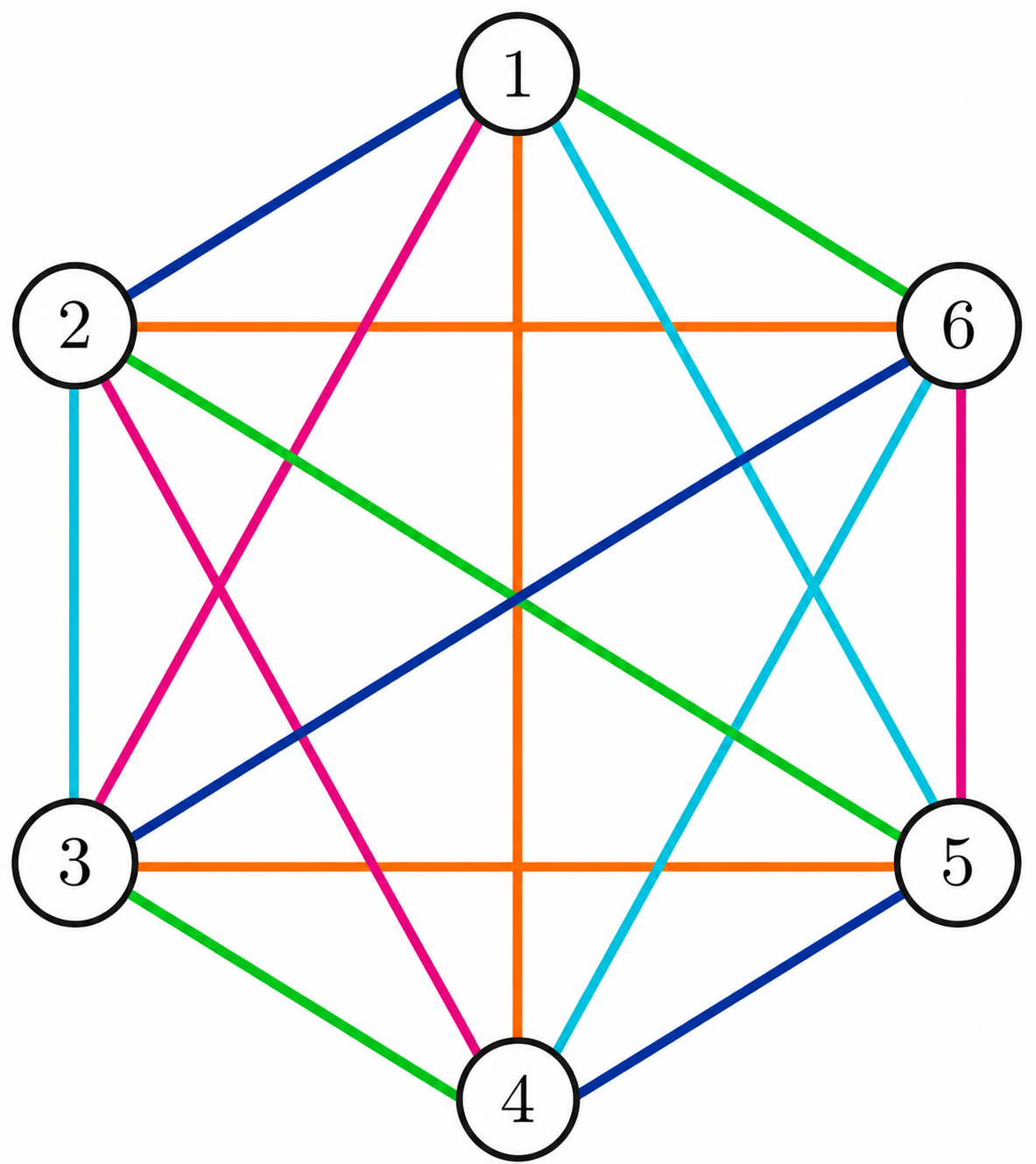}
\caption{\small 1-factorization of $K_6$.}\label{1-factor}
\end{wrapfigure}

Let us consider all sets of the form $\{x,y,z',t'\}$ where the edges $\{x,y\},\{z,t\}$ are of the same colour in the graph with coloured edges given in Figure~\ref{1-factor}. This figure plots a well-known 1-factorization of the complete graph on 6 vertices, i.e. the collection of perfect matchings (corresponding to edges of the same colour) whose disjoint union gives all edges.
The total number of the latter quadruples is $9\cdot 5=45$. In this $(12,4,4)$-code, $[6]$ and $\{1',\cdots,6'\}$ are two free zones.

These two free zones can be occupied by arbitrary maximum independent sets $I_1$ in $J_\pm([6],4)$ and $I_2$ in $J_\pm(\{1',2',3',4',5',6'\},4)$ respectively. Lemma~\ref{plug-in-J64} guarantees that this construction defines an independent set in $J_\pm(12,4)$. Since $\alpha(J_\pm(6,4))=2^4\cdot 3$, the cardinality of this set is $2\cdot 2^4\cdot 3+2^4\cdot 45=2^4\cdot 51$. From $\alpha(J(12,4))=51$ and Theorem~\ref{main-fact} we obtain that all such independent sets are maximum. Using computer-aided methods, such independent sets can all be listed and factorized modulo signed permutations. This gives us possible $1157$ non-isomorphic maximum independent sets.




\noindent {\bf Construction 2.} Best~\cite{Best1978} also describes another way to define the set of 45 quadruples taking two elements from both $[6]$ and $\{1',\cdots,6'\}$ with the property of $(12,4,4)$-code. We omit its description here. The total number of possible non-isomorphic maximum independent sets in $J_\pm(12,4)$ arising from this construction is $278$.

\noindent {\bf Construction 3.} This construction is based on the interchange of two elements taken from $[6]$ and $\{1',\cdots,6'\}$. The idea also originates in~\cite{Best1978} where it was employed to construct optimal $(12,4,4)$-codes. Let us adapt it to define maximum independent sets in $J_\pm(12,4)$.

Suppose that we have two isomorphic graphs with coloured edges and disjoint sets of vertices $\{a,b,c,d\}$ and $\{p,q,r,s\}$ as drawn on Figure~\ref{quadruples}. A set of quadruples combining equicoloured edges is the following set:

\[
\{\{a,b,p,q\},\quad
\{a,b,r,s\},\quad
\{a,c,p,r\},\quad
\{a,d,p,s\},
\]
\[
\{b,c,q,r\},\quad
\{b,d,q,s\},\quad
\{c,d,p,q\},\quad
\{c,d,r,s\}\}.
\]
We delete two quadruples from that list, $\{a,b,p,q\}$ and $\{c,d,r,s\}$, and replace them with $\{a,b,c,d\}$ and $\{p,q,r,s\}$, obtaining
\[
\mathcal{B}
= \{\{a,b,c,d\},\quad
\{a,b,r,s\},\quad
\{a,c,p,r\},\quad
\{a,d,p,s\},
\]
\[
\{b,c,q,r\},\quad
\{b,d,q,s\},\quad
\{c,d,p,q\},\quad
\{p,q,r,s\}\}.
\]
These are exactly the quadruples to be modified. Exchanging $d$ and $r$ in all those quadruple will result in a set
\begin{equation*}
\begin{split}
\mathcal{D}=\{a,b,c,r\},\quad
\{a,b,d,s\},\quad
\{a,c,p,d\},\quad
\{a,r,p,s\}
\\
\{b,c,q,d\},\quad
\{b,r,q,s\},\quad
\{c,r,p,q\},\quad
\{p,q,d,s\}.
\end{split}
\end{equation*}
\begin{wrapfigure}{r}{0.4\textwidth}
\centering
\includegraphics[width=0.4\textwidth]{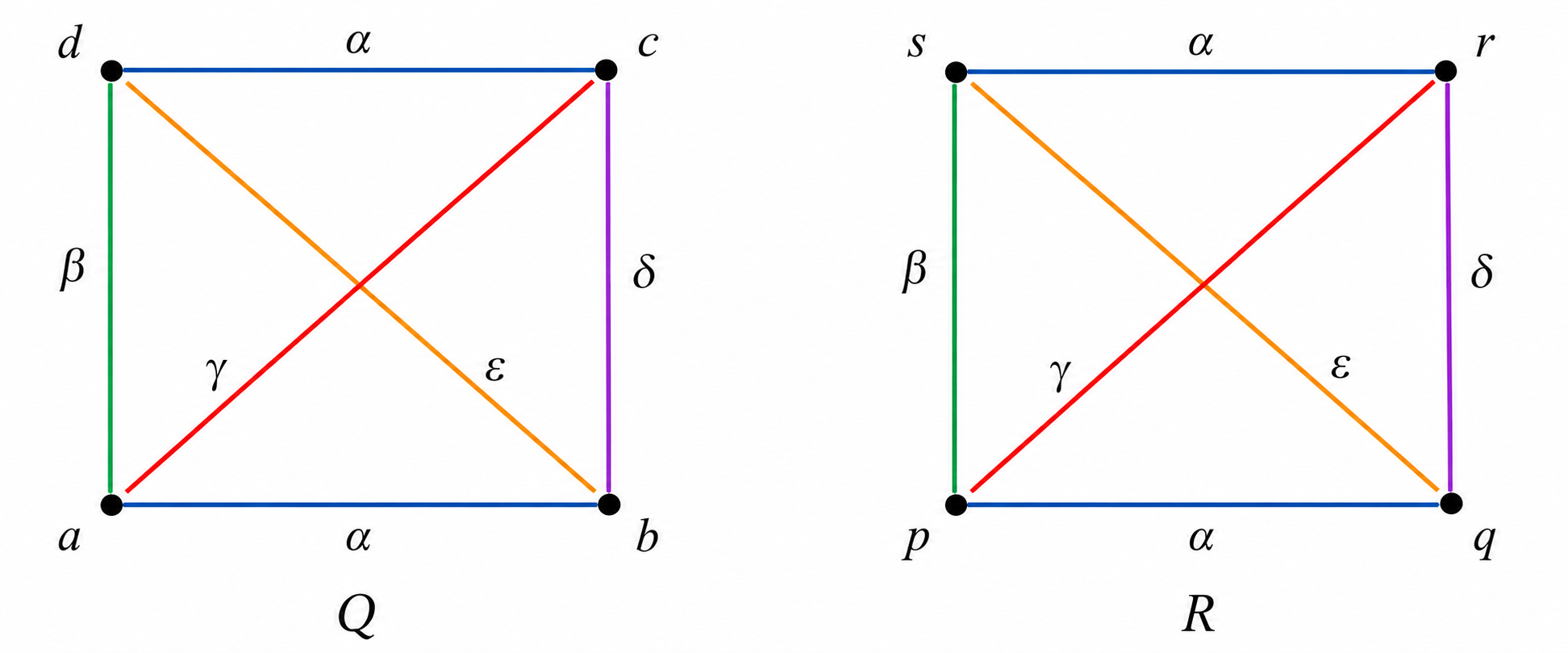}
\caption{\small Isomorphic quadruples.}\label{quadruples}
\end{wrapfigure}

The following lemma gives a construction of an independent set in $J_\pm(12,4)$.
\begin{lemma}
Let $A=\{a,b,c,d,e,f\}$, $B=\{p,q,r,s,t,u\}$ be two 6-element disjoint sets. We assume that quadruples $\{a,b,c,d\}$ and
$\{p,q,r,s\}$ define sets $\mathcal{B}$ and $\mathcal{D}$. Let
$M=\{\{a,b\},\{c,d\},\{e,f\}\}$
be a matching on $A$, and let
$N=\{\{p,q\},\{r,s\},\{t,u\}\}$
be a matching on $B$. Suppose that $M$ and $N$ are extended to $1$-factorizations of complete graphs with sets of vertices $A$ and $B$, respectively, and that these two $1$-factorizations are coloured by the same five colours consistent with Figure~\ref{quadruples}. Let $I_1\subset J_{\pm}(A,4)$ and $I_2\subset J_{\pm}(B,4)$ be independent sets such that
\[
T(\{a,b,c,d\})\subseteq I_1,
\qquad
T(\{p,q,r,s\})\subseteq I_2.
\]
Define
\[
\mathcal F
=
\left\{
\{x,y,z,t\}\mid
\{x,y\}\subset A,\ \{z,t\}\subset B,
\text{ and }\{x,y\},\{z,t\}\text{ have the same colour}
\right\}.
\]
and
\[
\begin{aligned}
X={}&
\bigl(I_1\setminus T(\{a,b,c,d\})\bigr)
\cup
\bigl(I_2\setminus T(\{p,q,r,s\})\bigr)\cup \bigcup_{S\in \mathcal{D}} T(S)\cup \bigcup_{C\in \mathcal F\setminus \mathcal B} T(C).
\end{aligned}
\]
Then $X$ is an independent set in $J_{\pm}(A\cup B,4)$.
\end{lemma}
\begin{proof} Since $I_1\setminus T(\{a,b,c,d\})$ and $I_2\setminus T(\{p,q,r,s\})$ are independent, we can omit considering intersection of two supports from $I_1\setminus T(\{a,b,c,d\})$ or two supports from $I_2\setminus T(\{p,q,r,s\})$.
So it is enough to prove that any other two distinct supports appearing in the
definition of $X$ intersect in at most two points . Indeed, if two
supports intersect in at most two coordinates, then the scalar product of
any two corresponding signed vectors is at most $2$.

First, since
$T(\{a,b,c,d\})\subseteq I_1$ and $I_1$ is independent, every vector
$x\in I_1\setminus T(\{a,b,c,d\})$
has support intersecting $\{a,b,c,d\}$ in at most two points. Otherwise
one could choose a vector $y\in T(\{a,b,c,d\})$ with the same signs on
three common coordinates, giving
$\langle x,y\rangle\ge 3$,
contradicting independence of $I_1$. Hence
\[
|\operatorname{supp}(x)\cap \{a,b,c,d\}|\le 2.
\]
Similarly, every vector $z\in I_2\setminus T(\{p,q,r,s\})$
satisfies
$|\operatorname{supp}(z)\cap \{p,q,r,s\}|\le 2$.

Every support in $\mathcal D$ intersects $A$ only inside
$\{a,b,c,d\}$ and intersects $B$ only inside $\{p,q,r,s\}$. Therefore
the previous observation implies that every support from
$I_1\setminus T(\{a,b,c,d\})$ or
$I_2\setminus T(\{p,q,r,s\})$ intersects every support from
$\mathcal D$ in at most two points.

Also, $C\in \mathcal F\setminus \mathcal B$ intersects every support from
$I_1\setminus T(\{a,b,c,d\})$ in at most two points. Indeed, $C\cap A$
has size two, while vectors from $I_1$ are supported entirely in $A$.
Similarly, $C$ intersects every support from
$I_2\setminus T(\{p,q,r,s\})$ in at most two points.

Now it remains to prove that distinct supports from $\bigcup_{S\in \mathcal{D}} T(S)\cup \bigcup_{C\in \mathcal F\setminus \mathcal B} T(C)$ intersect in at most 2 elements. This is equivalent to claiming that quadruples from $\mathcal{D}\cup (\mathcal F\setminus \mathcal B)$ correspond to a $(12,4,4)$-code. In fact, this is exactly the code that was constructed in~\cite{Best1978}. Let us reproduce here arguments from that paper, for completeness.

If the triple $\{x,y,z\}\subseteq \{a,b,c,d,p,q,r,s\}$ is such that $d,r\notin \{x,y,z\}$, then that triple is covered by   $\mathcal D$ if and only if it is covered by $\mathcal B$.
Other triples covered by $\mathcal D$ and  $\mathcal B$ are listed below adjacently:
\[
\begin{gathered}
abr,\ acr,\ bcr,
\qquad
abd,\ acd,\ bcd,
\\
abd,\ ads,\ bds,
\qquad
abr,\ ars,\ brs,
\\
acd,\ adp,\ cdp,
\qquad
acr,\ apr,\ cpr,
\\
arp,\ ars,\ rps,
\qquad
adp,\ ads,\ dps,
\\
bcd,\ bdq,\ cdq,
\qquad
bcr,\ bqr,\ cqr,
\\
brq,\ brs,\ qrs,
\qquad
bdq,\ bds,\ dqs,
\\
cpr,\ crq,\ pqr,
\qquad
cdp,\ cdq,\ dpq,
\\
dpq,\ dps,\ dqs,
\qquad
pqr,\ prs,\ qrs.
\end{gathered}
\]
It can be checked that these are exactly the same set of triples. So replacing $\mathcal B$ with $\mathcal D$ in the union $\mathcal{D}\cup (\mathcal F\setminus \mathcal B)$ maintains the property that any triple is covered only once. So it remains to prove that any triple is covered only once in $\mathcal F$.

Take two distinct supports
$C_1=\{x_1,y_1,z_1,t_1\}$, $C_2=\{x_2,y_2,z_2,t_2\}$
from $\mathcal F$. Each $C_i$ is the union of an
edge in $A$ and an edge in $B$, and the two edges have the same colour.
If $C_1$ and $C_2$ correspond to the same colour, then their $A$-edges
are either equal or disjoint, and their $B$-edges are either equal or
disjoint, because each colour class is a perfect matching. Since
$C_1\ne C_2$, we get
$|C_1\cap C_2|\le 2$.
If they correspond to different colours, then their $A$-edges share at
most one vertex and their $B$-edges share at most one vertex. Hence again
$|C_1\cap C_2|\le 2$.

Thus, $X$ is independent in
$J_{\pm}(A\cup B,4)$.
\end{proof}

There are exactly seven non-isomorphic maximum independent sets in
$J_{\pm}(6,4)$. Let
$I_i,\ I_j\subset J_{\pm}(6,4)$
be two canonical representatives, possibly equal. We regard $I_i$ as an
independent set on the coordinate set
$\{1,2,3,4,5,6\}$,
and $I_j$ as an independent set on the disjoint coordinate set
$\{1',2',3',4',5',6'\}$.

Suppose that
$T(\{a,b,c,d\})\subseteq I_i$, $T(\{p,q,r,s\})\subseteq I_j$.
Consider the matchings
\[
M=\{\{a,b\},\{c,d\},\{e,f\}\}
\]
on $\{1,2,3,4,5,6\}$ and
\[
N=\{\{p,q\},\{r,s\},\{t,u\}\}
\]
on $\{1',2',3',4',5',6'\}$, where
$\{e,f\}
=
[6]\setminus\{a,b,c,d\}$,
$\{t,u\}
=
\{1',2',3',4',5',6'\}\setminus\{p,q,r,s\}$.

Each of the matchings $M$ and $N$ extends to exactly two
$1$-factorizations of the corresponding complete graph. Choose one such
extension for $M$ and one such extension for $N$, and colour the two
$1$-factorizations by the same five colours in such a way that the
correspondence
\[
a\leftrightarrow p,\qquad
b\leftrightarrow q,\qquad
c\leftrightarrow r,\qquad
d\leftrightarrow s
\]
extends to a colour-preserving isomorphism between the two coloured complete
graphs.

Define
\[
\mathcal F=
\left\{
\{x,y,z,t\}\mid
\{x,y\}\subset [6],\
\{z,t\}\subset \{1',2',3',4',5',6'\},
\text{ and }\{x,y\},\{z,t\}\text{ have the same colour}
\right\}.
\]
We then define
\[
\begin{aligned}
I={}&
\bigl(I_i\setminus T(\{a,b,c,d\})\bigr)
\cup
\bigl(I_j\setminus T(\{p,q,r,s\})\bigr)\cup \bigcup_{S\in \mathcal{D}} T(S)\cup \bigcup_{C\in \mathcal F\setminus \mathcal B} T(C).
\end{aligned}
\]
By the previous lemma, $I$ is an independent set in $J_{\pm}(12,4)$.

Note that the canonical independent set $I_7$ does not contain a support
$A$ such that $T(A)$ is fully contained in $I_7$; therefore it cannot
be used in Construction~3. Each of the independent sets
$I_2,\ldots,I_6$
contains exactly one such support. In the unique classical independent set
$I_1$, all such supports are equivalent under the automorphism group of
$I_1$. Consequently, without loss of generality, we may always choose the
support $\{1,2,3,4\}$ in the case of $I_1$.

Thus, varying $i,j\in [6]$ and the choices of extensions of $M$ and
$N$ to $1$-factorizations produces a family of independent sets in
$J_{\pm}(12,4)$. By construction, there are
$6\cdot 6\cdot 4=144$
such labelled independent sets. Factoring this family by the action of the
signed permutation group yields a collection of $140$ pairwise
non-isomorphic independent sets.

Using computer-aided methods described in Section~\ref{computer-aided} we were able to find all possible maximum independent set in $J_{\pm}(12,4)$ up to signed permutations. The following theorem lists all possibilities.
\begin{theorem}\label{classification-Jpm12}
Every maximum independent set in $J_{\pm}(12,4)$ is obtained by one of the following constructions:
\begin{enumerate}
\item a maximum independent set obtained from the Construction~1 ($1157$ non-isomorphic sets);
\item a maximum independent set obtained from the Construction~2 ($278$ non-isomorphic sets);
\item a maximum independent set obtained from the Construction~3 ($140$ non-isomorphic sets).
\item  a classical independent set corresponding to one of the egotistic codes with
indices $14\le i\le 17$ in Best's list~\cite{Best1978} of optimal
$(12,4,4)$-codes.
\end{enumerate}
Overall, we have $1579=1157+278+140+4$ non-isomorphic maximum independent sets in $J_\pm(12,4)$.
\end{theorem}
\section{Characterization of maximum independent sets}
This section is dedicated to results of our computations of maximum independent sets in $J_\pm(n,4)$ for $5\leq n\leq 12$. For $n\leq 7$ we applied the igraph library~\cite{igraph} for a listing of all independent sets and distinguished isomorphic sets by comparing their full orbits. This approach drastically slows down already for $n=8$, so for $n\geq 8$ we used the following pipeline: (a) the computation of decoys by Algorithm~\ref{decoys} (using networkx library~\cite{SciPyProceedings_11}); (b) listing of all maximum independent sets compatible with the system of decoys (using MOSEK~\cite{mosek} for the branch-and-bound part); (c) the computation of the list of non-isomorphic maximum independent sets. The part (c) was needed only for $n=12$, as the output of the step (b) was already small enough in other cases.  In part (c) we encode each arrangement as a colored auxiliary graph and use the canonical labeling algorithm of nauty library~\cite{MCKAY201494} to compute a canonical certificate.

\subsection{The case of $n=6$}\label{case6}
Although 6 of 7 maximum independent sets are non-classical, they all can be understood as transformations of the classical one. For each independent set $I$ of this kind it is useful to characterize it using the set ${\rm supp}(I) = \{{\rm supp}(x)\mid x\in I\}$ and the edge set $E_I=\{[6]\setminus {\rm supp}(x)\mid x\in I\}$ of a graph $G_I=([6],  E_I)$. Also, by ${\rm profile}(I)$ we denote the multiset of numbers $\{|\{x\in I\mid {\rm supp}(x)=a\}|\}_{a\in {\rm supp}(I)}$. If $v\in \mathbb{R}^n$ and $v = (v_1,...,v_n)$, then for $a,b \in [n]$, we define
\[
v_{a\leftrightarrow b} = \pi(v),
\]
where $\pi$ is a permutation that swaps $a$ and $b$. 

\begin{enumerate}
\item The classical maximum independent set is given below.
\[
\begin{array}{c|c|cccccc}
{\rm supp}(x) & [6]\setminus {\rm supp}(x)
& x_1 & x_2 & x_3 & x_4 & x_5 & x_6\\
\hline
1234 & 56 & \pm1 & \pm1 & \pm1 & \pm1 & 0 & 0\\
1256 & 34 & \pm1 & \pm1 & 0 & 0 & \pm1 & \pm1\\
3456 & 12 & 0 & 0 & \pm1 & \pm1 & \pm1 & \pm1
\end{array}
\]
For this set we have, $G_I\simeq 3K_2$ and  
${\rm profile}(I)\simeq (16,16,16)$. The first 5 non-classical maximum independent sets can be understood as a swap of coordinates applied to vectors from  $T([6]\setminus \{5,6\})$, $T([6]\setminus \{3,4\})$ (that belong to the classical independent set) without touching $T([6]\setminus \{1,2\})$.
\item Let us apply the transformation
$$H_1(x) = \begin{cases}
        x_{3\leftrightarrow 6}, & \text{if } (x_1, x_2) =(1,1)\text{ and }x\in T([6]\setminus \{5,6\}) \cup T([6]\setminus \{3,4\}) \\
        x, & \text{otherwise }  \\
    \end{cases}.$$
on each vector of the classical independent set. Then, we obtain the first non-classical maximum independent set:
\[
\begin{array}{c|c|ccccccc}
{\rm supp}(x) & [6]\setminus {\rm supp}(x)
& x_1 & x_2 & x_3 & x_4 & x_5 & x_6 & \\
\hline
1234 & 56 & \alpha & \beta & \pm1 & \pm1 & 0 & 0 & (\alpha,\beta)\in\{(-1,-1),(-1,1),(1,-1)\}\\
1256 & 34 & \alpha & \beta & 0 & 0 & \pm1 & \pm1 & (\alpha,\beta)\in\{(-1,-1),(-1,1),(1,-1)\}\\
3456 & 12 & 0 & 0 & \pm1 & \pm1 & \pm1 & \pm1 &\\
1235 & 46 & 1 & 1 & \pm1 & 0 & \pm1 & 0 &\\
1246 & 35 & 1 & 1 & 0 & \pm1 & 0 & \pm1 &
\end{array}
\]
For this set we have, $G_I\simeq C_4\sqcup K_2$ and  
${\rm profile}(I)\simeq (16,12,12,4,4)$. 
\item Let us apply the transformation
$$H_2(x) = \begin{cases}
        x_{3\leftrightarrow 6}, & \text{if } x_1 =1\text{ and }x\in T([6]\setminus \{5,6\}) \cup T([6]\setminus \{3,4\}) \\
        x, & \text{otherwise }  \\
    \end{cases}.$$
on each vector of the classical independent set. Then, we obtain the second non-classical  maximum independent set:
\[
\begin{array}{c|c|cccccc}
{\rm supp}(x) & [6]\setminus {\rm supp}(x)
& x_1 & x_2 & x_3 & x_4 & x_5 & x_6\\
\hline
1234 & 56 & -1 & \pm1 & \pm1 & \pm1 & 0 & 0\\
1256 & 34 & -1 & \pm1 & 0 & 0 & \pm1 & \pm1\\
3456 & 12 & 0 & 0 & \pm1 & \pm1 & \pm1 & \pm1\\
1235 & 46 & 1 & \pm1 & \pm1 & 0 & \pm1 & 0\\
1246 & 35 & 1 & \pm1 & 0 & \pm1 & 0 & \pm1
\end{array}
\]
For this set we have, $G_I\simeq C_4\sqcup K_2$ and  ${\rm profile}(I)\simeq (16,8,8,8,8)$.
\item Let us apply the transformation
$$H_3(x) = \begin{cases}
        x_{3\leftrightarrow 6}, & \text{if } (x_1,x_2) =(1,-1)\text{ and }x\in T([6]\setminus \{5,6\}) \cup T([6]\setminus \{3,4\}) \\
        x_{3\leftrightarrow 5}, & \text{if } (x_1,x_2) =(1,1)\text{ and }x\in T([6]\setminus \{5,6\}) \cup T([6]\setminus \{3,4\}) \\
        x, & \text{otherwise }  \\
    \end{cases}.$$
on each vector of the classical independent set. Then, we obtain the third non-classical maximum independent set:
\[
\begin{array}{c|c|cccccc}
{\rm supp}(x) & [6]\setminus {\rm supp}(x)
& x_1 & x_2 & x_3 & x_4 & x_5 & x_6\\
\hline
1234 & 56 & -1 & \pm1 & \pm1 & \pm1 & 0 & 0\\
1256 & 34 & -1 & \pm1 & 0 & 0 & \pm1 & \pm1\\
3456 & 12 & 0 & 0 & \pm1 & \pm1 & \pm1 & \pm1\\
1235 & 46 & 1 & -1 & \pm1 & 0 & \pm1 & 0\\
1246 & 35 & 1 & -1 & 0 & \pm1 & 0 & \pm1\\
1236 & 45 & 1 & 1 & \pm1 & 0 & 0 & \pm1\\
1245 & 36 & 1 & 1 & 0 & \pm1 & \pm1 & 0
\end{array}
\]
For this set we have, $G_I\simeq  K_4\sqcup K_2$ and  ${\rm profile}(I)\simeq (16,8,8,4,4,4,4)$.
\item Let us apply the transformation
$$H_4(x) = \begin{cases}
        x_{3\leftrightarrow 6}, & \text{if } x_1\ne x_2\text{ and }x\in T([6]\setminus \{5,6\}) \cup T([6]\setminus \{3,4\}) \\
        x, & \text{otherwise }  \\
    \end{cases}.$$
on each vector of the classical independent set. Then, we obtain the fourth non-classical  maximum independent set:
\[
\begin{array}{c|c|cccccc}
{\rm supp}(x) & [6]\setminus {\rm supp}(x)
& x_1 & x_2 & x_3 & x_4 & x_5 & x_6\\
\hline
1234 & 56 & \alpha & \alpha & \pm1 & \pm1 & 0 & 0\\
1256 & 34 & \alpha & \alpha & 0 & 0 & \pm1 & \pm1\\
3456 & 12 & 0 & 0 & \pm1 & \pm1 & \pm1 & \pm1\\
1235 & 46 & \alpha & -\alpha & \pm1 & 0 & \pm1 & 0\\
1246 & 35 & \alpha & -\alpha & 0 & \pm1 & 0 & \pm1
\end{array}
\]
For this set we have, $G_I\simeq  C_4\sqcup K_2$ and  ${\rm profile}(I)\simeq (16,8,8,8,8)$.
\item Let us apply the transformation
$$H_5(x) = \begin{cases}
        x_{3\leftrightarrow 6}, & \text{if } x_1\ne x_2\text{ and }x\in T([6]\setminus \{5,6\}) \cup T([6]\setminus \{3,4\}) \\
        x_{3\leftrightarrow 5}, & \text{if } (x_1, x_2)=(1,1)\text{ and }x\in T([6]\setminus \{5,6\}) \cup T([6]\setminus \{3,4\}) \\
        x, & \text{otherwise }  \\
    \end{cases}.$$
on each vector of the classical independent set. Then, we obtain the fifth non-classical  maximum independent set:
\[
\begin{array}{c|c|cccccc}
{\rm supp}(x) & [6]\setminus {\rm supp}(x)
& x_1 & x_2 & x_3 & x_4 & x_5 & x_6\\
\hline
1234 & 56 & -1 & -1 & \pm1 & \pm1 & 0 & 0\\
1256 & 34 & -1 & -1 & 0 & 0 & \pm1 & \pm1\\
3456 & 12 & 0 & 0 & \pm1 & \pm1 & \pm1 & \pm1\\
1235 & 46 & \alpha & -\alpha & \pm1 & 0 & \pm1 & 0\\
1246 & 35 & \alpha & -\alpha & 0 & \pm1 & 0 & \pm1\\
1236 & 45 & 1 & 1 & \pm1 & 0 & 0 & \pm1\\
1245 & 36 & 1 & 1 & 0 & \pm1 & \pm1 & 0
\end{array}
\]
For this set we have, $G_I\simeq  K_4\sqcup K_2$ and  ${\rm profile}(I)\simeq (16,8,8,4,4,4,4)$.
\item To obtain the last type of maximum independent set we have to transform elements of all three $T([6]\setminus \{5,6\})$, $T([6]\setminus \{3,4\})$ and $T([6]\setminus \{1,2\})$, so it is easier to describe it directly.

Partition the set ${\mathbb F}_3^2\setminus \{(0,0)\}$ into two 4-element subsets, one consisting of $\{(0,-1),(-1,0),(0,1),(1,0)\}$ (i.e. all pairs with one zero) and another consisting of $\{(-1,1),(-1,-1),(1,-1),(1,1)\}$ (the rest).  Define any bijective mapping $f: \{(0,-1),(-1,0),(0,1),(1,0)\}\to \{3,4,5,6\}$. Now we define a mapping $m: {\mathbb F}_3^2\setminus \{(0,0)\}\to {\mathbb F}_2^4$ for any $(a,b)\in \{(0,-1),(-1,0),(0,1),(1,0)\}$ by
$$
m(a,b)=\mathds{1}_{\{3,4,5,6\}\setminus \{f(a,b)\}},
$$
and  for any $(a,b),(c,d)\in \{(0,-1),(-1,0),(0,1),(1,0)\}$  such that $(a,b)\oplus (c,d)\ne (0,0)$, by
$$
m((a,b)\oplus (c,d))=m(a,b)\oplus m(c,d).
$$
Then, the maximum independent set is defined by
$$
I = \left\{\{(a,b)\}\times T(m(a,b))\mid (a,b)\in  {\mathbb F}_3^2\setminus \{(0,0)\}\right\}.
$$
If we choose the bijection $f$ as $\{(0,-1)\to 3,(-1,0)\to 4,(0,1)\to 5,(1,0)\to 6\}$, the table of the independent set is given below.
\[
\begin{array}{c|c|cccccc}
{\rm supp}(x) & [6]\setminus {\rm supp}(x)
& x_1 & x_2 & x_3 & x_4 & x_5 & x_6\\
\hline
2456 & 13
& 0 & -1 & 0 & \pm1 & \pm1 & \pm1
\\
1356 & 24
& -1 & 0 & \pm1 & 0 & \pm1 & \pm1
\\
2346 & 15
& 0 & 1 & \pm1 & \pm1 & 0 & \pm1
\\
1345 & 26
& 1 & 0 & \pm1 & \pm1 & \pm1 & 0
\\
1245 & 36
& -1 & 1 & 0 & \pm1 & \pm1 & 0
\\
1234 & 56
& -1 & -1 & \pm1 & \pm1 & 0 & 0
\\
1236 & 45
& 1 & -1 & \pm1 & 0 & 0 & \pm1
\\
1256 & 34
& 1 & 1 & 0 & 0 & \pm1 & \pm1
\end{array}
\]
For this set we have, $G_I\simeq  K_{3,3}\setminus e$ (oriented edges are $\{1,4,6\}\times \{2,3,5\}\setminus \{(1,2)\}$) and  ${\rm profile}(I)\simeq (8,8,8,8,4,4,4,4)$.
\end{enumerate}

\subsection{The cases of $n=7$ and $n=9$}\label{steiner-triple-leave}
\begin{wrapfigure}{r}{0.35\textwidth}
\centering
\includegraphics[width=0.35\textwidth]{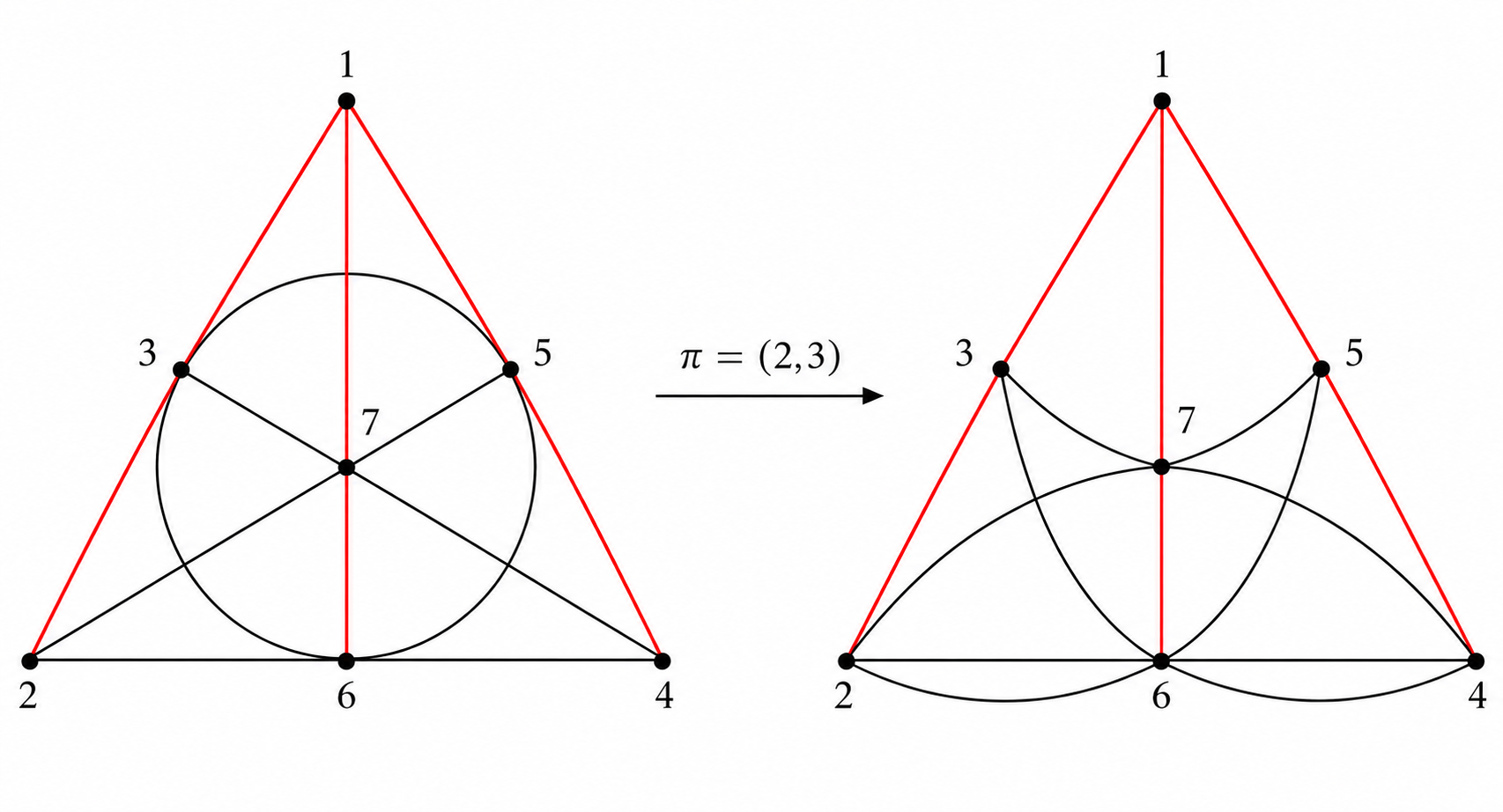}
\caption{\small The transformation of Fano plane under $\pi=(2,3)$.}\label{1-factor}
\end{wrapfigure}
There are only two non-isomorphic independent sets for $n=7$. The first
one is clearly classical. For the second one, we compared the corresponding
kissing arrangement with the arrangement
$\mathcal{C}_{\rm 7\text{-}126b}$, provided as supplementary material
to~\cite{cohn2011rigidity}, and verified that the two arrangements are
isomorphic. The latter arrangement was originally discovered
by~\cite{Conway1995}.

Note that the Conway--Sloane construction is a special case of the
independent set whose existence is guaranteed by
Lemma~\ref{conway}. 
Indeed, we show that there exist two distinct optimal
$(7,4,4)$-codes $C$ and $C'$ such that
\[
\{S\mid \mathds{1}_S\in C,\ 1\notin S\}
=
\{S\mid \mathds{1}_S\in C',\ 1\notin S\}.
\]

Recall that the codewords of an optimal $(7,4,4)$-code are complements of
the triples of a Steiner system $S(2,3,7)$. The latter can be represented
geometrically by the lines of the Fano plane (see
Figure~\ref{1-factor}). Under this correspondence, quadruples not containing
the element $1$ correspond precisely to triples containing $1$, shown in
red in the figure.

Now apply the permutation
$\pi=(2,3)$. 
This produces another Fano plane, distinct from the original one, while
preserving the collection of red lines. Consequently, the corresponding
optimal $(7,4,4)$-code $C'$ is distinct from $C$, yet satisfies
$\{S\mid \mathds{1}_S\in C,\ 1\notin S\}
=
\{S\mid \mathds{1}_S\in C',\ 1\notin S\}$.
By Lemma~\ref{conway}, the pair $C,C'$ defines a non-classical maximum
independent set in $J_\pm(7,4)$. This yields precisely the second
non-isomorphic solution. 

Analogously, for $n=9$ we found two non-isomorphic maximum independent sets: the classical one and the one whose kissing arrangement is isometric to the kissing arrangement found by~\cite{cohn2026variationsfivedimensionalspherepackings}. Again, let us show that the second independent set is the one described in Lemma~\ref{conway}.

Let us identify $[9]$ with $[3]^2$. Then, an optimal $(9,4,4)$-code is defined by the following set of quadruples
\[
C
=
\left\{
([3]\setminus\{a\})\times([3]\setminus\{b\})
\mid
a,b\in[3]
\right\}
\]
\[
\qquad\qquad\cup
\left\{
(\{a\}\times([3]\setminus\{b\}))
\cup
(([3]\setminus\{a\})\times\{b\})
\mid
a,b\in[3]
\right\}.
\]
Set
$u=(1,1)$.
Define
\[
\pi
=
\bigl((2,2)\ (3,3)\bigr)
\bigl((2,3)\ (3,2)\bigr),
\]
fixing all other points of $[3]^2$.
Equivalently, $\pi$ acts on the lower-right $2\times2$ subgrid by the
central reflection with center $(2.5,2.5)$.

\begin{lemma}
The permutation $\pi$ satisfies
\[
\pi(C_{\overline{u}})
=
C_{\overline{u}},
\]
but
\[
\pi(u)=m,
\qquad
\pi(C_{\bar u})=C_{\bar m},
\qquad
\pi(C)\neq C.
\]
\end{lemma}

\begin{proof}
First we describe $C_{\overline u}$. The code consists of two types of blocks: $R_{a,b}=([3]\setminus\{a\})\times([3]\setminus\{b\})$,
and
$X_{a,b}=
(\{a\}\times([3]\setminus\{b\}))
\cup
(([3]\setminus\{a\})\times\{b\})$.

A rectangle $R_{a,b}$ avoids $u=(1,1)$ if and only if
$a=1 \quad\text{or}\quad b=1$. 
Thus the rectangle blocks in $C_{\overline u}$ are
\[
R_{1,1},R_{1,2},R_{1,3},R_{2,1},R_{3,1}.
\]
The permutation $\pi$ permutes the rectangle blocks of $C_{\overline u}$ as follows:
\[
R_{1,1}\leftrightarrow R_{1,1},
\]
Define $\pi$ by
\[
R_{1,2}\leftrightarrow R_{1,3},
\]
with $(0,2)$ fixed. Then clearly
\[
R_{2,1}\leftrightarrow R_{3,1}.
\]
Therefore the set of rectangle blocks in $C_{\overline u}$ is preserved.

Similarly, a cross $X_{a,b}$ avoids $u=(1,1)$ if and only if either $(a,b)=(1,1)$,
or $a,b\in\{2,3\}$.
Thus the cross blocks in $C_{\overline u}$ are
\[
X_{1,1},X_{2,2},X_{2,3},X_{3,2},X_{3,3}.
\]
By construction, all those blocks do not change under $\pi$.
Therefore the set of cross blocks in $C_{\overline u}$ is also preserved.
Consequently,
$\pi(C_{\overline u})=C_{\overline u}$.

It remains to show that $\pi$ does not preserve $C$. Consider the block
\[
S=R_{3,3}
=
\{(1,1),(1,2),(2,1),(2,2)\}.
\]
Clearly $S\in C$. But
\[
\pi(S)
=
\{(1,1),(1,2),(2,1),(3,3)\}.
\]
This set is neither a rectangle $R_{a,b}$ nor a cross $X_{a,b}$. Hence $\pi(C)\neq C$ and lemma proved.
\end{proof}
So, $C$ and $\pi(C)$ are exactly two optimal $(9,4,4)$-codes needed for Lemma~\ref{conway}.
\subsection{The cases of $n=8,10,14,16$}\label{allsteiner}
According to Corollary~\ref{classical-steiner}, all independent sets are permutations of the classical one.  For $n=14$ and $n=16$, we did not perform explicit computations. The
corresponding entries in the table were instead obtained theoretically from
the fact that all independent sets arise from egotistic Steiner systems. The
numbers of such systems are known to be $4$ and $1054163$, respectively.

\subsection{The cases of $n=5,11$}\label{case5mod6}
The behaviour for $n\equiv 5 \pmod 6$ appears to be substantially more
complicated than in the other congruence classes. In particular, the number
of non-isomorphic independent sets grows dramatically. For instance, in the
case $n=5$ our computations produced $336$ pairwise non-isomorphic
independent sets. The situation for $n=11$ is considerably more involved.

In the case $n=11$, the computational complexity becomes prohibitively
large already at the level of decoys. More precisely, there are $11343$
decoys to consider, compared to only $7476$ in the case $n=12$. Moreover,
already the first decoy for $n=11$ produces more than $2000$ candidate
solutions, whereas a typical decoy for $n=12$ produces only about $10$
solutions. Consequently, our current computational resources are insufficient
to fully process the case $n=11$.

We believe that this phenomenon is related to the fact that optimal
$(n,4,4)$-codes for $n\equiv 5 \pmod 6$ are considerably less rigid than
in the other congruence classes. Such irregular behaviour is well known in
the theory of constant-weight codes~\cite{bok:MW}.

\section{Isomorphism of maximum independent sets}\label{Isomorphism}
The following theorem is instrumental in our computer-aided counting of non-isometric kissing arrangements coming from independent sets of $J_{\pm}(n,4)$.
\begin{theorem}\label{isomorphism}
Let $n\geq 9$, and let $I,I'\subseteq V(J_{\pm}(n,4))$ be maximum
independent sets 
and
\[
\mathcal K=I\cup E,\qquad \mathcal K'=I'\cup E,
\]
where $E=\{\pm 2e_1,\dots,\pm 2e_n\}$.
Suppose that $U\in O(n)$ is an isometry such that
\[
U(\mathcal K)=\mathcal K'.
\]
Then $U(E)=E$, i.e. every isometry between the two kissing arrangements
is a signed permutation.
\end{theorem}

\begin{remark} The statement of the latter theorem is not true for $n=8$. Indeed, if $I$ is a maximum independent set in $J_\pm(8,4)$, then the size of $E\cup I$ is $16+16\cdot 14=240$, which is equal to the size of the root system of $E_8$. Due to the uniqueness of an optimal kissing arrangement in ${\mathbb R}^8$ (up to isometry)~\cite{Bannai_Sloane_1981}, $E\cup I$ is isometrical to the root system of $E_8$. So it should preserve all automorphisms of $E_8$, including reflection with respect to any  ${\mathbf v}\in I$, i.e.
$$
{\mathbf x}\to {\mathbf x}- \frac{2{\mathbf v}^\top {\mathbf x}}{{\mathbf v}^\top {\mathbf v}}{\mathbf v}.
$$
It can be easily checked that the latter tranformation cannot be given as a signed permutation.
\end{remark}
To prove Theorem~\ref{isomorphism} we need a couple of lemmas.
\begin{lemma}\label{needed-bound}
For every $n>8$,
\[
A(n,4,4)>1+A(n-4,4,4)+3(n-4).
\]
\end{lemma}

\begin{proof}
We use the known exact Johnson value~\cite{JOHNSON1972109,10.1007/s10623-014-0001-2}
\[
A(n,4,4)=
\begin{cases}
\left\lfloor
\dfrac n4
\left\lfloor
\dfrac{n-1}{3}
\left\lfloor
\dfrac{n-2}{2}
\right\rfloor
\right\rfloor
\right\rfloor,
& n\not\equiv 0 \pmod 6,\\[2ex]

\left\lfloor
\dfrac n4
\left\lfloor
\dfrac{n-1}{3}
\left\lfloor
\dfrac{n-2}{2}
\right\rfloor
\right\rfloor
-1
\right\rfloor,
& n\equiv 0 \pmod 6,
\end{cases}
\]
where $\left\lfloor x\right\rfloor$ denotes the largest integer not more than $x$.

For $n=9,10$ we directly check:
\[
A(9,4,4)=18,\qquad A(5,4,4)=1,
\]
so
\[
18>1+1+15=17,
\]
and,
\[
A(10,4,4)=30,\qquad A(6,4,4)=3,
\]
so
\[
30>1+3+18=22.
\]
Now assume $n\ge 11$. Since
\[
\left\lfloor \frac{n-2}{2}\right\rfloor\ge \frac{n-3}{2},
\]
we have
\[
A(n,4,4)
\ge
\frac{n}{4}
\left(
\frac{(n-1)(n-3)}{6}-1
\right)-2.
\]
Also,
\[
A(n-4,4,4)
\le
\frac{(n-4)(n-5)(n-6)}{24}.
\]
Therefore
\[
\begin{aligned}
&A(n,4,4)-A(n-4,4,4)-3(n-4)-1 \\
&\ge
\frac{n(n-1)(n-3)}{24}
-\frac {n}{4}
-2
-\frac{(n-4)(n-5)(n-6)}{24}
-3(n-4)-1.
\end{aligned}
\]
Simplifying gives
\[
A(n,4,4)-A(n-4,4,4)-3(n-4)-1
\ge
\frac{11n^2-149n+336}{24}.
\]
The largest root of the latter quadratic polynomial is $\frac{149 + \sqrt{7417}}{22}\approx 10.7$,
hence
\[
A(n,4,4)>1+A(n-4,4,4)+3(n-4),
\]
for $n\geq 11$.
\end{proof}

\begin{lemma}\label{1-intersection}
Let $n>8$, and let $C\subseteq \{0,1\}^n$ be an optimal binary
$(n,4,4)$-code. Then for every $x\in C$ there exists 
$x'\in C$ such that
\[
|{\rm supp}(x)\cap {\rm supp}(x')|=1.
\]
\end{lemma}

\begin{proof}
Fix $x=\mathds{1}_S\in C$, and suppose, for contradiction, that no block of $C$
intersects $S$ in exactly one point. Since distinct blocks of $C$ have
intersection at most $2$, every $\mathds{1}_T\in C\setminus\{\mathds{1}_S\}$ satisfies
\[
|S\cap T|\in\{0,2\}.
\]
Put $m=n-4$.
First, the blocks disjoint from $S$ form a binary $(m,4,4)$-code on
$\{0,1\}^{[n]\setminus S}$. Hence there are at most
\[
A(m,4,4)=A(n-4,4,4)
\]
such blocks.

Second, consider blocks $T$ with $|S\cap T|=2$. Each such block is of
the form
\[
T=P\cup R,
\]
where $P\in {S\choose 2}$ and $R\in {[n]\setminus S\choose 2}$.
Fix a pair $P\in {S\choose 2}$. If
\[
T_1=P\cup R_1,\qquad T_2=P\cup R_2
\]
are two distinct blocks of $C$, then $R_1$ and $R_2$ must be disjoint.
Indeed, if $R_1\cap R_2\neq\varnothing$, then
\[
|T_1\cap T_2|\geq |P|+1=3,
\]
contradicting the $(n,4,4)$ condition.

Therefore, for a fixed pair $P\subset S$, there are at most
\[
\left\lfloor \frac{n-4}{2}\right\rfloor
\]
blocks intersecting $S$ exactly in the pair $P$. Since $S$ has
$\binom{4}{2}=6$ pairs, the total number of blocks intersecting $S$ in
two points is at most
\[
6\left\lfloor \frac{n-4}{2}\right\rfloor
\leq 3(n-4).
\]
Combining the two estimates, we get
\[
|C|
\leq
1+A(n-4,4,4)+3(n-4).
\]
However, for $n>8$, Lemma~\ref{needed-bound} gives us
\[
A(n,4,4)>
1+A(n-4,4,4)+3(n-4).
\]
This contradicts the optimality of $C$.

Therefore for every $\mathds{1}_S\in C$ there
exists $\mathds{1}_{S'}\in C$ such that
$|S\cap S'|=1$.
\end{proof}

\begin{proof}[Proof of Theorem	~\ref{isomorphism}]
First note that elements of $E$ can be distinguished from elements of $I$  by the following geometric relationship with other vectors:
\[
E=
\left\{
p\in\mathcal K:
p^\top q\notin\{-1,1\}
\text{ for all } q\in\mathcal K\setminus\{p\}
\right\}.
\]
Indeed, for every $p\in E$ and every $q\in \mathcal K\setminus\{p\}$,
it is straightforward to check $p^\top q\notin\{-1,1\}$. 

Conversely, let $x\in I$. Let $Q$ be the hyperoctant containing $x$.
Then, by Theorem~\ref{main-fact}, $I\cap Q$ is, after eliminating signs, an optimal
binary $(n,4,4)$-code. Let $S=\operatorname{supp}(x)$. By Lemma~\ref{1-intersection}, there exists a vector
$y\in I\cap Q$ such that
\[
|\operatorname{supp}(x)\cap \operatorname{supp}(y)|=1.
\]
Since $x$ and $y$ lie in the same hyperoctant, their signs agree on
their common coordinate. Hence,
\[
x^\top y=1.
\]
Thus, every vector $x\in I$ has another vector $y\in\mathcal K$ with
inner product $1$.

The same characterization holds for $E\subseteq\mathcal K'$. Since
$U$ is orthogonal, it preserves all inner products. Therefore it
preserves the above intrinsic property. Hence, $U(E)=E$.
\end{proof}

\bibliographystyle{IEEEtran}
\bibliography{lit}

\end{document}